\documentclass[11pt,letterpaper]{article}

\usepackage[margin=1in]{geometry}
\usepackage{url}
\usepackage{enumitem}
\usepackage{amsmath}
\usepackage{amssymb}
\usepackage{verbatim}
\usepackage{amsthm}
\usepackage{algorithm}
\usepackage{algorithmic}
\usepackage{bm}
\usepackage{color}
\usepackage{graphicx}
\usepackage{multirow}  
\usepackage[dvipsnames,usenames]{xcolor}
\usepackage[colorlinks=true,urlcolor=blue,linkcolor=RoyalBlue,citecolor=OliveGreen]{hyperref}
\usepackage{enumitem}

\usepackage{tikz}
\usepackage[numbers,sort]{natbib}

\newtheorem{definition}{Definition}[section]

\newtheorem{theorem}{Theorem}[section]
\newtheorem{corollary}[theorem]{Corollary}

\newtheorem{lemma}[theorem]{Lemma}

\newcommand{\E}{\mathbb{E}}

\renewcommand{\O}{\mathcal{O}}

\DeclareMathOperator*{\argmax}{argmax}

\DeclareMathOperator*{\OPT}{OPT}
\DeclareMathOperator*{\ALG}{ALG}

\newcounter{note}[section]

\newcommand{\polylog}{\ensuremath{\text{polylog}}}

\renewcommand{\vec}[1]{\bm {#1}}
\newcommand*\samethanks[1][\value{footnote}]{\footnotemark[#1]}

\newcommand{\jinyan}[1]{\textcolor{Green}{(Jinyan: #1)}}
\newcommand{\xiangning}[1]{\textcolor{blue}{(Xiangning: #1)}}

\DeclareMathOperator*{\pdf}{Pdf}
\DeclareMathOperator*{\pr}{Pr}

\title{Learning Optimal Reserve Price against Non-myopic Bidders}

\author{
	Zhiyi Huang\thanks{The University of Hong Kong. Email: \{zhiyi, jyliu, xnwang\}@cs.hku.hk~.}
	\and
	Jinyan Liu\samethanks
	\and 
	Xiangning Wang\samethanks
}

\date{}

\begin{document}

\begin{titlepage}
	\thispagestyle{empty}
	
	\maketitle
	
	\begin{abstract}
	We consider the problem of learning optimal reserve price in repeated auctions against non-myopic bidders, who may bid strategically in order to gain in future rounds even if the single-round auctions are truthful.
	Previous algorithms, e.g., empirical pricing, do not provide non-trivial regret rounds in this setting in general.
	We introduce algorithms that obtain small regret against non-myopic bidders either when the market is large, i.e., no bidder appears in a constant fraction of the rounds, or when the bidders are impatient, i.e., they discount future utility by some factor mildly bounded away from one.
	Our approach carefully controls what information is revealed to each bidder, and builds on techniques from differentially private online learning as well as the recent line of works on jointly differentially private algorithms.
\end{abstract}

\end{titlepage}

\section{Introduction}
\label{sec:intro}

The problem of designing revenue-optimal auctions based on data has drawn much attention in the algorithmic game theory community lately.
Various models have been studied, notably, the sample complexity model~\cite{cole2014sample, dughmi2014sampling, huang2015making, morgenstern2015pseudo, morgenstern2016learning, roughgarden2016ironing, devanur2016sample, gonczarowski2017efficient} and, very recently, the online learning model~\cite{bubeck2017online}.
However, these existing works all implicitly assume that bidders are myopic in the sense that they will faithfully report their valuations as long as the mechanism used in each round is truthful, without considering how their current bids may affect 1) the choices of mechanisms and 2) the behaviors of other bidders in future rounds in which they may also participate in.
What happens in the presence of non-myopic bidders?

\begin{center}
	\fcolorbox{lightgray}{lightgray}{
		\parbox{0.96\textwidth}{%
			\textbf{Example.~}
			Suppose a seller has a fresh copy of the good for sale every day, where its value for any bidder is bounded between $0$ and $1$.
			The seller sets a price at the beginning of each day.
			Then, a bidder (say, randomly chosen from a large yet finite pool of potential bidders) arrives and submits a bid: If the bid is higher than the price of the day, he gets the item and pays the price.
			Further, suppose the seller adopts the solution proposed by the sample complexity literature and decides to set the price at $0.5$ on day $1$, and in each of the following days to use the empirical price, i.e., the best fixed price w.r.t.\ the bids in previous days.
		}
	}
\end{center}


If bidders are myopic, bidding truthfully is a dominant strategy since the mechanism on each day is effectively posting a take-it-or-leave-it price.
As a result, the seller will be able to converge to the optimal fixed price w.r.t.\ to the pool of potential bidders.

If bidders are non-myopic, however, their strategies are more intriguing. 
A bidder may underbid whenever the empirical price of the day (which is deterministic) is higher than his value, inducing the same result of not winning in the current round but leads to lower future prices compared with truthful bid.
By the same reasoning, even if the bidder wants to win the current copy of the good, he will not bid truthfully; instead, he will submit a bid that equals the current price.
Due to the strategic plays of non-myopic bidders, the seller may in fact converge to price close to $0$.

The take-away message of this example is that directly applying existing algorithms to scenarios where bidders are non-myopic could be a disaster in terms of revenue.
Hence, we ask:
%
\medskip
\begin{quote}
	\em
	Are there learning algorithms that work well even in the presence of non-myopic agents?
\end{quote}
\medskip
In particular, can we extend the online learning model to allow non-myopic bidders, and  design algorithms that provably have small regret against the best fixed auction?

\subsection{Our Results and Techniques}

Our main contribution is a positive answer to the above question, subject to one of the following two assumptions: 
Either the bidders are impatient in the sense that they discount utility in future rounds by some discount factor (mildly) bounded away from $1$ (\emph{impatient bidders}), or any bidder comes in only a small fraction of the rounds (\emph{large market}). 
The assumptions are relative natural, and are further necessary: If the same bidder appears everyday without discounting future utility, no algorithms can guarantee non-trivial regret (e.g., Theorem 3 of \citet{amin2013learning}). 


\paragraph{Single-bidder.}
Let us first consider the case in the example, where only a single bidder comes on each day.
We show the following:

\bigskip

\noindent
\textbf{Informal Theorem 1.~} 
{\em 
	For any $\alpha > 0$, our online pricing algorithm has regret at most  $\alpha T$ against non-myopic bidders when $T \ge \tilde{\O}(\alpha^{-4})$, and either impatient bidders or large market holds.
}

\bigskip

This is equivalent to a sub-linear $\tilde{\O}(T^{3/4})$ regret bound.
Here, we omit in the big-O notation a term that depends on either the discount factor or the maximum number of rounds that a bidder can appear.
The same remark also applies to the other informal theorem in this subsection.

Typical mechanism design approach may seek to design online learning mechanisms such that truthful biddings form an equilibrium (or even better, are dominating strategies) even if bidders are non-myopic, and to get small regret given truthful bids.
However, designing truthful mechanisms that learn in repeated auctions seems beyond the scope of existing techniques.

We take an alternative path by relaxing the incentive property: We aim to ensure that bidders would only submit ``reasonable'' bids within a small neighborhood of the true values.
Note that the notion of regret is robust to small deviations of bids. 
Applying an online learning algorithm on ``reasonable'' bids (instead of truthful bids) results in only a small increase in the regret.

To explain how to achieve the relaxed incentive property, first consider why a bidder would lie.
Since the single-round auctions are truthful, a bidder can never gain in the current round by lying.
Lying is preferable only if the future gain outweighs the current loss (if any).
The seller's algorithm in the example, i.e., using empirical prices, suffers on both ends.
On one hand, lying has no cost when the bid and true value are both greater (or both smaller) than the price.
On the other hand, the current bid has huge influence on future prices; in particular, the first bid dictates the price in the second round.

We will design online auctions such that (1) deviating too far from the true value in the current round is always costly, and (2) the influence of the current bid on future prices/utilities is bounded.
Achieving the first property turns out to be easy.
Note that lying has a cost whenever the price falls between the bid and true value.
On each day, with some small probability our mechanisms pick the price randomly to ensure that the price has a decent probability to fall between the value and any ``unreasonable'' bid that deviates a lot.

The second property may seem trivial through the following incorrect argument: 
Online learning algorithms, e.g., multiplicative weight, are intrinsically insensitive to the bid on any single day and, thus, satisfy the second property.
The argument incorrectly assumes subsequent bids will remain the same regardless of the current bid, omitting that they are controlled by strategic bidders and, thus, are affected by the current bid through its influence on subsequent prices.
We use an implementation of the follow-the-perturbed-leader algorithm based on the tree-aggregation technique~\cite{agarwal2017price} from differential privacy.
Note that due to the above reasoning, differential privacy does not imply the second property.
Nevertheless, we show that the algorithm in fact satisfies a slightly stronger guarantee than differential privacy, which is sufficient for our proof.

Our techniques further allow us to get essentially the same bound even in the bandit setting, i.e., seller observes if the bidder buys the copy but not his bid. We only make a sketch in the main text and the full proof 
is shown in Appendix~\ref{sec:single-bandit}. 

\paragraph{Multi-bidder.}
Our approach can be extended to obtain positive results with $n > 1$ bidders and $m > 1$ copies of the good per round, with some extra ingredients we shall explain shortly.

\bigskip

\noindent
\textbf{Informal Theorem 2.~} 
{\em 
	For any $\alpha > 0$, our algorithm runs an approximate version of Vickrey with an anonymous reserve price on each day with regret at most $\alpha m T$ against the best fixed reserve price if (1) $T \ge \tilde{O}(\frac{n}{m \alpha^{4.5}})$, (2) $m \ge \tilde{O}(\frac{\sqrt{n}}{\alpha^3})$, (3) either impatient bidders or large market holds.
}

\bigskip


Simply running a follow-the-perturbed-leader algorithm with tree-aggregation as in the single-bidder case does not work in the multi-bidder setting because a bidder's current bid can now affect other bidders' subsequent bid through the allocations and payments in the current round.
%

To make our approach work in the multi-bidder setting, we need two more ingredients.
First, we need to refine our model to control what information are revealed to the bidders in the single-round auctions.
At the end of each round, each bidder can observe his own outcome, i.e., whether he wins a copy of the good and his payment.
At any round, a bidder's (randomized) bid can depend on his own bids and outcomes in previous rounds, but not those of the other bidders.
The information structure plays a crucial role in the argument of bidders' incentives.

Then, it boils down to bounding the influence of a bidder's bid on other bidders' outcomes in the same round.
This is exactly the main feature of jointly differentially private mechanisms developed in a series of recent papers~\cite{kearns2014mechanism, hsu2016private, hsu2016jointly, huang2018better}.
After choosing the reserve price in each round, we run an approximate version of Vickrey with reserve as follows.
First, run the jointly private algorithm of \citet{hsu2016private} to get a set of roughly $m$ candidate bidders and an approximate Vickrey price.
Then, for each candidate bidder, offer a take-it-or-leave-it price that equals the chosen reserve price or the approximate Vickrey price, whichever is higher.
The joint privacy of single-round auctions together with the previous argument on the learning process bound how much a bidder's current bid can affect his future utility.
Finally, the approximation guarantees of the joint private algorithm ensure that the revenue loss is bounded compared with running Vickrey with the same reserve.

\subsection{Related Work}
\label{sec:related}

There is a vast literature on (revenue) optimal auction design.
We discuss only the most related works in the single-parameter setting. 
\citet{myerson1981optimal} showed that optimal auctions are (ironed) virtual surplus maximizers.
Hence, there are a continuum of possible Myerson auctions since each of them is characterized by $n$ mapping from values to virtual values.
\citet{devanur2016sample} observed that to get a $1 - \alpha$ approximation it suffices to consider a finite set of Myerson-type auctions, each of which is given by $n$ mappings from discretized values to discretized virtual values that are multiples of $\alpha$ between $0$ and $1$.
We take the same discretized viewpoint in this paper.
If the bidders' value distributions are i.i.d.\ and regular, the optimal auction has a simpler form of a Vickrey auction with a reserve price that equals the monopoly price of the distribution.
Further, even if distributions are not i.i.d.\ a Vickrey auction with a suitable reserve still gets a constant approximation~\cite{hartline2009simple}.



\citet{cole2014sample} proposed to study the sample complexity of optimal auctions, and showed upper and lower bounds polynomial in (the inverse of) the error term $\alpha$ and the number of bidders $n$ for single-parameter problems. 
\citet{huang2015making} characterized the sample complexity up to a log factor for the special case of selling to a single bidder.
\citet{morgenstern2015pseudo, morgenstern2016learning, roughgarden2016ironing, devanur2016sample, gonczarowski2017efficient} significantly improved the upper bounds for both Myerson-type auctions and simpler auction formats in the single-parameter setting.
\citet{bubeck2017online} revisited the problem in an online-learning model and introduced algorithms that simultaneously achieve near optimal regret against arbitrary bidder values, improving a series of previous results \cite{blum2004online, blum2005near, kleinberg2003value}, and near optimal sample complexity if values are drawn from a underlying distribution.
\citet{dughmi2014sampling} proved that exponentially many samples are needed in general for multi-parameter problems.
\citet{cai2017learning} considered multi-item auctions and showed that polynomially many samples is enough for getting a constant multiplicative approximation plus an $\alpha$ additive loss.
These works implicitly assume myopic bidders so either previous bids are truthful when previous auctions are truthful, or an approximation of the prior distribution can be estimated from the bids in non-truthful previous auctions through econometric techniques~\cite{nekipelov2015econometrics}.
This paper takes a more proactive approach of directly investigating how to design the learning process together with the auction format of each round so that we can extract meaningful information from the bids even if bidders are non-myopic.

Our results build on two lines of work in differential privacy, namely, differentially private online learning algorithms, and jointly differentially private algorithms.
\citet{agarwal2017price} introduced an $(\varepsilon,\delta)$-differentially private algorithm with regret $\tilde{O}(\sqrt{T}+{\sqrt{K}}/{\varepsilon})$ for the full information setting (a.k.a.\ the expert problem), and regret $\tilde{O}(\sqrt{TK}/\varepsilon)$ for the bandit setting, improving the results in a series of previous work \cite{tossou2016algorithms, thakurta2013nearly, jain2012differentially, jain2014near}. 
Here, $T$ is the number of rounds and $K$ is the number of experts/arms.
Independently, \citet{tossou2017achieving} discovered essentially the same result for the bandit setting.
\citet{kearns2014mechanism} revisited the notion of differential privacy by Dwork et al.~\cite{dwork2006calibrating, dwork2006differential} and gave a relaxation called joint differential privacy which can be applied to many combinatorial problems for which no differentially private algorithm gets non-trivial approximation.
\citet{hsu2016private} introduced the billboard lemma which serves as the cornerstone of subsequent works on joint differential privacy, and developed a jointly private algorithm for matching and combinatorial auctions with gross substitute valuations.
\citet{hsu2016jointly} proposed a general framework for designing jointly private algorithms for a large family of convex optimization problems.
Finally, \citet{huang2018better} introduced a jointly private algorithm that is optimal up to log factors for packing problems.

\paragraph{Previous Models of Repeated Auctions with Non-myopic Bidders.}
Alternative models of repeated auctions against non-myopic bidders have been considered in the literature.
\citet{devanur2015perfect} and \citet{immorlica2017repeated} studied the equilibria of repeated sales when seller cannot commit to a pricing strategy and, thus, must play according to a perfect Bayesian equilibrium.
In contrast, we adopt the standard assumption that the seller as a mechanism designer can commit to a strategy upfront.
Further, their models assume the same bidder comes every day with the same value drawn from a prior upfront, while our model assumes no prior, allows different bidders on different days, and allows the same bidder to have distinct values on different days.
\citet{amin2014repeated, amin2013learning} considered a stochastic version of the online learning model in this paper with the same bidder coming every day, and proposed algorithms with sub-linear regrets.
\citet{mohri2014optimal} showed how to get better regret bounds in the special case when bidder has the same value on different days. 
We stress that our model is more general as it assumes no prior and allows different bidders on different days.
Allowing different bidders, in particular, brings a lot of challenges.



%

\paragraph{Previous Applications of Differential Privacy in Mechanism Design.}
Although differential privacy has been applied to mechanism design before our work, its role in previous work (e.g, \cite{mcsherry2007mechanism, nissim2012approximately, epasto2018incentive, hsu2016jointly, hsu2016private, huang2018better}) is fundamentally different from that in ours.
First, following a proposal by \citet{mcsherry2007mechanism}, most previous work used differential privacy to achieve approximate incentive compatibility in the sense that misreporting cannot increases a bidder's utility by more than an $\epsilon$ amount.
Further, such mechanisms can be coupled with a strictly truthful mechanism to achieve exact incentive compatibility in some specific problems~\cite{nissim2012approximately}.
In contrast, our work makes use of some techniques from differential privacy (rather than the concept itself) to control the influence of a bidder's bid in any single round on his utility in future rounds.
Then, we build on this property to bound how much a bidder may deviate from his/her true value in equilibria (instead of the amount of incentive to deviate as in previous approaches).
Hence, our approach is conceptually different from the previous ones.

Second, previous work generally used differential privacy to design one-shot mechanisms, while our work considers repeated auctions.
Characterizing bidder's behaviors is notoriously hard, a single bidder's deviation in a single round may have the cascading effect of changing the bids of all bidders in subsequent rounds. 
To this end, we propose to use joint differential privacy as a mean to control information dissemination and, consequently, bidders' behaviors in future rounds in repeated auctions.
This is a novel application of joint differential privacy to our knowledge.

In concurrent and independent work, \citet{epasto2018incentive} considered incentive-aware learning and used differential privacy to control the amount of a bidder's deviation from his true value using an approach similar to ours. 
However, their work focused on a one-shot interaction environment while ours consider repeated auctions. 
The results are therefore  incomparable.

\section{Preliminary}

\subsection{Single-bidder Model}
Let there be a seller who has a fresh copy of the good for sale every day for a total of $T$ days.
Exactly one bidder comes on each day, but the same bidder may show up in multiple days.
We assume that a bidder can come on at most $\tau$ days for some $\tau \le T$.
Consider the following interactions between the seller and bidders.
On each day $t \in [T]$:
\begin{enumerate}[itemsep=0mm]
	\item Seller sets a price $p_t \in [0, 1]$ as a (randomized) function of previous bids $b_1, \dots, b_{t-1}$.
	\item A bidder arrives with value $v_t \in [0, 1]$ and submits a bid $b_t \in [0, 1]$ as a (randomized) function of his value $v_t$, his bids and auction outcomes in the previous rounds that he participates in.
	\item Seller observes the bid $b_t$ but not the value $v_t$.
	\item Bidder receives the good and pays $p_t$ if $b_t \ge p_t$; nothing happens otherwise.
\end{enumerate}
Here, it is crucial to assume that a bidder does not observe the bids and auction outcomes of the rounds in which he does not participate.

\paragraph{Rational Bidders.}
%
A bidder's utility in a single round is quasi-linear, namely, $v_t - p_t$ if he gets the good and $0$ otherwise.
For some discount factor $\gamma \in [0, 1]$, a bidder discounts future utility by at least $\gamma$ and seeks to maximize the sum of discounted utilities.
For example, suppose a bidder comes on days $t_1$, $t_2$, and $t_3$.
When the bidder considers his strategy on day $t_1$, he would sum up his utilities from all three days, discounting future utility on day $t_2$ by at least $\gamma$, and that on day $t_3$ by at least $\gamma^2$.
If $\gamma = 0$, it becomes the model with myopic bidders.
If $\gamma = 1$, bidders simply seek to maximize the sum of their utilities.
Note that we do not assume that the values of the same bidder must be the same on different days (although they could be).

\paragraph{Seller.}
As a mechanism designer, the seller can commit to a mechanism, i.e., fixing the (randomized) pricing functions $p_1, \dots, p_T$ upfront.
Hence, we shall interpret the $T$-round interactions as a game among the bidders, with the seller designing (part of) the rules.
The seller aims to maximize revenue, i.e., the sum of the prices payed by the bidders over all $T$ rounds, denoted as $\ALG$.

We adopt the standard regret analysis of online learning and compare $\ALG$ with the optimal fixed price in hindsight, namely, $\max_{p \in [0, 1]} p \cdot \sum_{t\in[T]} \mathbf{1}_{v_t \ge p}$.\footnote{Note that a bidder's best strategy against a fixed price is truthful bidding.}
We denote by $\OPT(\{v_t\}_t)$ the revenue of the best fixed price w.r.t.\ a given sequence of values $\{v_t\}_t$.
The \emph{regret} of the algorithm is therefore %
\[
\OPT(\{v_t\}_t) - \ALG ~.
\]
We will further split the regret into two parts as follows in our analysis:
\[
\OPT(\{v_t\}_t) - \ALG = \underbrace{\OPT(\{v_t\}_t) - \OPT(\{b_t\}_t)}_{\text{game-theoretic regret}} + \underbrace{\OPT(\{b_t\}_t) - \ALG}_{\text{learning regret}}
\]

\paragraph{Assumptions.}
We consider instances that satisfy one of the following two assumptions.
The hardness result by~\citet{amin2013learning} implies that no non-trivial regret is possible if neither holds.
\begin{itemize}[itemsep=0mm, parsep=0mm, topsep=1mm]
	\item \textbf{Large-market:~} No bidder participates in a significant portion of rounds, i.e., $\tau = o(T)$.
	\item \textbf{Impatient bidder:~} $\gamma$ is (mildly) bounded away from $1$, i.e., $\frac{1}{1 - \gamma} = o(T)$.
\end{itemize}
%




\subsection{Multi-bidder Model}

The model extends straightforwardly to the setting when multiple bidders come and the seller has multiple copies for sale every day.
We sketch the model below and highlight a few key assumptions.
The seller has $m$ fresh copies of the good for sale every day for a total of $T$ days.
$n$ buyers come on each day and a bidder can show up on at most $\tau \le T$ days.
On each day $t \in [T]$:
\begin{enumerate}[itemsep=0mm]
	\item Seller chooses an auction $M_t$ based on previous bid vectors $\vec{b}_1, \dots, \vec{b}_{t-1}$.
	\item $n$ bidders arrive with values $\vec{v}_t \in [0, 1]^n$ and observe $M_t$.
	\item Each bidder $i$ submits a bid $b_{ti} \in [0, 1]$ as a (randomized) function of the auction $M_t$, his value $v_{ti}$, and his values, bids, and auction outcomes in the previous rounds that he participates in.
	\item Seller observes the bids $\vec{b}_t$ but not the values $\vec{v}_t$, and runs $M_t$ with the bids.
	\item Each bidder observes his own outcome given by $M_t$.
\end{enumerate}
Again, a bidder cannot observe the auction outcomes of the rounds in which he does not participate.
Further, a bidder cannot observe the auction outcomes of the other bidders, i.e., who gets a copy of the good and how much they pay, even if he participates in that round.
Both assumptions on the information structure are crucial for our incentive argument.

Bidders are rational and seek to maximize the sum of their discounted utilities.
Seller aims to maximize the total revenue of all $T$ rounds, denoted as $\ALG$.
The benchmark, however, is not the revenue of the best fixed arbitrary auction.
Instead, we will compare $\ALG$ with the revenue of the best fixed auction within a certain family, in this paper, the family of Vickrey auctions with an anonymous reserve price.
We denote this benchmark as $\OPT( \{ \vec{v}_t \}_t )$. 

In online learning, the algorithm usually has the same strategy space as the offline benchmark. 
Our model, however, allows the seller to use auctions outside the family of benchmark auctions.
We stress that our algorithm uses this flexibility only to implement approximate versions of Vickrey auctions with reserves.
Hence, the benchmark is still meaningful.
It is an interesting open question if we can achieve positive results using only Vickrey auctions with reserves.

One can ask the same learning question about other families of auctions, e.g., learning the best anonymous Myerson-type auctions as in \citet{roughgarden2016ironing} or the best Myerson-type auctions as in \citet{devanur2016sample}.
Extending the techniques in this paper to handle these more complicated auction formats against non-myopic bidders is another interesting future direction.

\subsection{Differential Privacy Preliminaries}
\label{sec:dp-definitions}
Our techniques rely on the notion of differential privacy by Dwork et al.~\cite{dwork2006calibrating, dwork2006differential}, and its relaxation called joint differential privacy by \citet{kearns2014mechanism}.

\begin{definition}[Differential Privacy~\cite{dwork2006calibrating, dwork2006differential}]	
	An algorithm $A: \mathcal{C}^n \mapsto \mathcal{R}$ is $(\epsilon,\delta)$-differentially private if for all $S \subseteq \mathcal{R}$ and for all neighboring datasets $D,D'\in \mathcal{C}^n$ that differ in one entry, there is $\Pr[A(D) \in S] \le e^\epsilon \Pr[A(D') \in S]+\delta$.
	%
\end{definition}
\begin{definition}[Joint Differential Privacy~\cite{kearns2014mechanism}]
	An algorithm $A: \mathcal{C}^n \to \mathcal{R}^n$ is $(\varepsilon,\delta)$-\textnormal{jointly differentially private} if for any $i$, any $D,D'\in \mathcal{C}^n$ differ only in the $i$-th entry, and any $S\in \mathcal{R}^{n-1}$, there is $\pr [A(D)_{-i}\in S]\le e^{\varepsilon}\pr [A(D')_{-i}\in S]+\delta$.
\end{definition}
Differential privacy is immune to post-processing. 
Running an algorithm on the output of a differentially private algorithm without further access to the data is still differentially private.
\begin{lemma}[Post-processing~\cite{dwork2014algorithmic}]\label{lemma:post-processing}
	Let $f:\mathcal{C} \to \mathcal{R}$ be a randomized algorithm that is $(\epsilon,\delta)$-differentially private. Let $g:\mathcal{R} \to \mathcal{R}'$ be an arbitrary randomized mapping. Then $g \circ f : \mathcal{C} \to \mathcal{R}'$ is $(\epsilon,\delta)$-differentially private.
\end{lemma}

It also guarantees that privacy parameter scales gracefully when integrating different private subroutines into a single algorithm.
This is formalized as the following composition theorem.
\begin{lemma}[Composition Theorem~\cite{dwork2010boosting}]
	\label{lemma:composition-theorem}
	Let $f_i:\mathcal{C} \to \mathcal{R}$ be an $(\epsilon_i,\delta_i)$-differentially private algorithm for $i \in [m].$ Then $f_{[m]} = (f_1(D), \cdots ,f_m(D))$ is $(\sum_{i=1}^m\epsilon_i, \sum_{i=1}^m\delta_i)$-differentially private.
\end{lemma}
There are some basic mechanisms to achieve differential privacy. 
In particular, we make use of the Gaussian mechanism, which can be applied to numerical problems. 
\begin{lemma}[Gaussian Mechanism~\cite{dwork2006calibrating}]
	\label{lemma:gaussian-mechanism}
	For a function $f:\mathcal{C} \to R^m$ with $\ell_2$ sensitivity $\Delta_2(f)=\max \lVert f(\mathcal{C})-f(\mathcal{C}') \rVert_2$, Gaussian mechanism outputs $f(\mathcal{C}) + Z$, where $Z \sim {N}(0,\sigma^2)^m$ with $\sigma \ge \sqrt{2 \ln(1.25/\delta)} \Delta_2(f)/\epsilon$ is $(\epsilon,\delta)$-differentially private.
\end{lemma}

\section{Single-bidder Case}
\label{sec:single-sketch}
Following the treatment of \citet{bubeck2017online}, we restrict our attentions to prices that are multiples of $\alpha$ and treat each of such prices as an expert in an online learning problem.
Let $K = \frac{1}{\alpha} + 1$ denote the number of decretized prices.
Consider an expert problem with $K$ experts with the $i$-th expert corresponding to price $(i-1)\alpha$.
We will assume without loss that bids fall into the discretized price set.
The bid on day $t$, $b_t$, induces a gain vector $\vec{g}_t$ such that the gain of the $i$-th expert is $(i-1)\alpha$ if $b_t \ge (i-1)\alpha$ and $0$ otherwise.
That is,
%
$
\vec{g}_t = \big( 0, \alpha, 2\alpha, \dots, \lfloor \frac{b_t}{\alpha} \rfloor \alpha, 0, \dots, 0 \big)
$.
Further denote the accumulative gain vector up to time $t$ as $\vec{G}_t = \sum_{j \in [t]} \vec{g}_j$.

\begin{theorem}
	\label{thm:single-bidder}
	For any $\alpha > 0$, there is an online algorithm with regret $\O(\alpha T)$ when $T \ge \tilde{\O} \big( \tau \alpha^{-4} \big)$ under the large market assumption, or $T \ge \tilde{\O} \big( \frac{\alpha^{-4}}{1-\gamma}  \big)$ under the impatient bidder assumption.
\end{theorem}


In this section, we first present a simplified algorithm that gets regrets bound under stronger assumptions of $T \ge \tilde{\O} \big( \tau \alpha^{-4.5} \big)$ and $T \ge \tilde{\O} \big( \frac{\alpha^{-4.5}}{1-\gamma}  \big)$ from Subsection~\ref{subsection:normal_tree} to Subsection~\ref{subsection:full-one-fold-regret}.
We then give a more complexed algorithm and prove the regret bounds from Subsection~\ref{subsection:single-autcion-2Dtree} to Subsection~\ref{subsection:full-two-fold-regret}.

We also extend the algorithm to posted pricing, i.e. the bandit setting, where the seller only knows whether the item is sold or not and the payment instead of the exact bid value. We present a sketch in Subsection \ref{subsection:bandit_sketch}, and the full proof is in the Appendix.

\subsection{(Simplified) Algorithm}
\label{subsection:normal_tree}
\paragraph{Tree-aggregation.}
The simplified algorithm is a privacy-preserving version of the followed-the-perturbed-leader algorithm based on the tree-aggregation technique \cite{dwork2010differential, chan2011private}.
Since our analysis needs to make use of the structure of the algorithm, 
it is worthwhile to devote a few paragraphs to formally define the tree-aggregation subroutine.

Suppose we have $T$ elements (the experts' gains) and need to calculate the cumulative sum of elements from $1$ to $t$ for any $t \in [T]$ in a differentially private manner.
The na\"ive approach simply calculates the cumulative sums and add, say, Gaussian noise, to each of them.
Since an element may appear in all $T$ cumulative sums, the noise scale is $\tilde{\O}(\frac{\sqrt{T}}{\epsilon})$ by a standard argument.
Instead, the tree-aggregation technique calculates $T$ partial sums such that (1) each element appears in at most $\log T$ partial sums, and (2) each cumulative sum is the sum of at most $\log T$ partial sums.
This technique significantly reduces the noise scale to $\tilde{\O}(\frac{1}{\epsilon})$.

Next, we explain how to design the partial sums.
Consider any $t \in [T]$ with binary representation $(t_{\log T} \dots t_1 t_0)_2$, i.e., $t = \sum_{j=0}^{\log T} t_j \cdot 2^j$. 
Let $j_t$ be the lowest non-zero bit.
The $t$-th sum is over 
\[
\Lambda_t = \big\{ t - 2^{j_t} + 1, t - 2^{j_t} + 2, \dots, t-1, t \big\} ~.
\]
To compute the sum of the first $t$ elements, it suffices to sum up the following sets of partial sum obtained by removing the non-zero bits of $t$ one by one from lowest to highest:
\[
\textstyle
\Gamma_t = \big\{ t' \ne 0 : t' = t - \sum_{j = 0}^{h-1} t_j 2^j, h = 0, 1, \dots, \log T \big\}
\]
Then, we have $[t] = \cup_{j \in \Gamma_t} \Lambda_j$.

For example, suppose $t = 14 = 1 \cdot 2^1 + 1 \cdot 2^2 + 1 \cdot 2^3$.
Then, we have $\Lambda_t = \{13, 14\}$ and $\Gamma_t = \{14,12,8\}$.
The tree-aggregation subroutine is given in Algorithm~\ref{alg:tree-aggregation}.

\begin{algorithm}[t]
	\caption{~Tree-aggregation}
	\label{alg:tree-aggregation}
	\begin{algorithmic}[1]
		\STATE \textbf{input:~} 
		dimension $K$, gain vector $\mathbf{g}_t \in [0, 1]^K$ of each round $t$, noise scale $\sigma$
		\STATE \textbf{internal states:~} noisy partial sum $\vec{A}_t$ for $t \in [T]$
		\STATE \textbf{initialize:~} $\vec{A}_t = \vec{\mu}_t$ for all $t \in T$, with $\mu_{tj}$'s i.i.d.\ from normal distribution $N(0, \sigma^2)$.
		\FOR{t = 1, 2, \dots, T} 
		\STATE Receive $\vec{g}_t$ as input.
		\STATE Let $\vec{A}_j = \vec{A}_j + \vec{g}_t$ for all $j$ s.t.\ $t \in \Lambda(j)$. 
		\STATE \textbf{Output} $\tilde{\vec{G}}_t = \sum_{j \in \Gamma_t} \vec{A}_j + \vec{\nu}_t$, with $\nu_{tj}$'s i.i.d.\ from $N \big(0, (\log T + 1 - |\Gamma_t|) \sigma^2 \big)$. \label{step:first}
		\ENDFOR
	\end{algorithmic}
\end{algorithm}

The usual description of tree-aggregation (e.g., \cite{agarwal2017price}) computes the partial sum $\vec{A}_t$ in one-shot at step $t$ while ours considers $\vec{A}_t$'s as internal states that are maintained throughout the algorithm. 
Both descriptions result in the same algorithm but ours is more convenient for our proof.

\begin{lemma}[\citet{jain2012differentially}]
	\label{lem:tree-aggregation-privacy-previous}	
	The {final values} of the internal states $\vec{A}_t$'s are $(\epsilon, \delta = \frac{\epsilon}{T})$-differentially private with $\sigma = \frac{8 \sqrt{K}}{\varepsilon} \log T \sqrt{\ln \frac{\log T}{\delta}}$.
\end{lemma}

We need a slightly stronger version that is a simple corollary.

\begin{lemma}
	\label{lem:tree-aggregation-privacy}
	Fixed any $t_0 \in [T]$, the values of the internal states $\vec{A}_t$'s after round $t_0$ are $(\epsilon, \delta = \frac{\epsilon}{T})$-differentially private for bids on or before day $t_0$ with $\sigma = \frac{8 \sqrt{K}}{\varepsilon} \log T \sqrt{\ln \frac{\log T}{\delta}}$.
\end{lemma}

\begin{proof}
	The values of $\vec{A}_t$'s after time $t_0$ are effectively the values if subsequent gain vectors are all zero.
	Hence, the lemma follows as a corollary of Lemma~\ref{lem:tree-aggregation-privacy-previous}.
\end{proof}

\begin{lemma}
	\label{lem:tree-aggregation-noise}
	$\tilde{G}_{tj} - G_{tj}$ follows ${N}(0, (\log T + 1) \sigma^2)$ for any $t \in [T]$ and any $j \in [K]$.
\end{lemma}

\paragraph{Online Pricing Algorithm.}
The algorithm (Algorithm~\ref{alg:single-bidder-sketch}) is a variant of the privacy-preserving online learning algorithm of \citet{agarwal2017price}.
It uses tree-aggregation as a subroutine for maintaining an noisy version of the cumulative gains of each price.\footnote{Note that the different expression of setting price for two-fold tree-aggregation is due to we consider prices/experts in decreasing order in this case.}
On each day $t$, with some small probability it picks the price randomly; otherwise, it picks the price with the largest noisy cumulative gain in previous days, i.e, the largest entry of $\tilde{\vec{G}}_{t-1}$.

\begin{algorithm}[t]
	\caption{~Online Pricing (Single-bidder Case)}
	\label{alg:single-bidder-sketch}
	\begin{algorithmic}[1]
		\STATE \textbf{parameters:~} 
		regret parameter $\alpha$, $K = \frac{1}{\alpha} + 1$, privacy parameter $\epsilon$, $\delta = \frac{\epsilon}{T}$.
		\STATE \textbf{initialize} tree-aggregation (Alg.~\ref{alg:tree-aggregation}) with $\sigma = \frac{8 \sqrt{K}}{\varepsilon} \log T \sqrt{\ln \frac{\log T}{\delta}}$.
		
		(Or two-fold tree-aggregation (Alg.~\ref{alg:tree-aggregation-2d}) with $\sigma = \frac{8\log T \log K}{\varepsilon}  \sqrt{\ln(\frac{ \log K \log T}{\delta})}$.)
		\FOR {$t=1,\ldots,T$}
		\STATE With probability $\alpha$, pick $j \in [K]$ uniformly at random.\label{step:add-uniform} 
		\STATE Otherwise, pick $j$ that maximizes $\tilde{\vec{G}}_{(t-1)j}$.
		\STATE Set price $(j-1) \alpha$. (Or set price $(K-j) \alpha$ when using two-fold tree-aggregation (Alg.~\ref{alg:tree-aggregation-2d}).)
		\STATE Observe bid $b_t$ and, thus, the gain vector $\vec{g}_t$; update tree-aggregation.
		\ENDFOR
	\end{algorithmic}
\end{algorithm}

\subsection{Bound Learning Regret}
\label{sec:single-autcion-regret-2d}
The next lemma follows from our Lemma~\ref{lem:tree-aggregation-noise} and Theorem 8 of \citet{abernethy2014online}.

\begin{lemma}
	\label{lem:ftpl-regret}
	Consider running Algorithm~\ref{alg:single-bidder-sketch} without step~\ref{step:add-uniform}.
	Then, the learning regret w.r.t.\ the best fixed discretized price is at most $\O \big(\sqrt{\log K} (\sigma \sqrt{\log T} + \frac{T}{\sigma \sqrt{\log T}} ) \big)$.
\end{lemma}

\begin{corollary}
	\label{cor:single-bidder-learning-regret-sketch}
	The learning regret of Algorithm~\ref{alg:single-bidder-sketch} is 
	$\O \big( \sqrt{\log K} (\sigma \sqrt{\log T} + \frac{T}{\sigma \sqrt{\log T}}) + \alpha T \big)$.
\end{corollary}
\begin{proof}
	Running Algorithm~\ref{alg:single-bidder-sketch} with step~\ref{step:add-uniform} increases the regret by at most $\alpha T$.
	Further note that the regret w.r.t.\ to the best fixed discretized price differs from the actual regret by at most $\alpha T$.
\end{proof}

\subsection{Bounding Game-theoretic Regret: Stability of Future Utility}
\begin{lemma}[Stability of Future Utility]
	\label{lem:single-bidder-stability}
	For any bidder and any day $t$ on which he comes, the bidder's equilibria utilities in subsequent rounds in the subgames induced by different bids on day $t$ differ by at most an $e^\epsilon$ multiplicative factor plus a $\delta T$ additive factor.
\end{lemma}

This is the main argument of our approach.
First consider a seemingly intuitive yet incorrect proof. 
Since tree-aggregation is differentially private, the online pricing algorithm is also private treating the bid on each day as an entry of the dataset.
The lemma holds because changing the bid on a day leads to a neighboring dataset and, thus, the probability of any subset of future outcomes does not change much. 
This is incorrect because subsequent bids are controlled by strategic bidders.
Changing the bid on a day does not result in a neighboring dataset in general.\footnote{Although other bidders do not see what happens on day $t$ and, thus, will employ the same strategies in subsequent days, the actual bids are affected also by the seller's subsequent prices, which is affected by the bid on day $t$.}

\begin{proof}
	We shall abuse notation and refer to the bidder that comes on day $t$ as bidder $t$.
	Fixed any bidder $t$'s strategy for subsequent days (after round $t$).
	That is, fixed the (randomized) bidding function on any subsequent day $t'$ as a function only on his bids and auction outcomes between day $t$ and $t'$ (exclusive).
	Let us consider the resulting utilities for bidder $t$ in the subgames induced by two distinct bids on day $t$.
	Note that the other bidders' subsequent strategies will be the same in the subgames since they cannot observe what happens on day $t$.
	We shall interpret the execution of the online pricing algorithm, i.e., the algorithm together with the bidders' strategies in subsequent rounds, after round $t$ as a post-processing on the internal states of the tree-aggregation algorithm after day $t$ and, thus, is $(\epsilon, \delta)$-differentially private due to Lemma~\ref{lem:tree-aggregation-privacy}.
	Therefore, the utilities of any fixed subsequent strategy of bidder $t$ in the two subgames differ by at most an $e^\epsilon$ multiplicative factor plus a $\delta T$ additive factor.
	The lemma then follows by the equilibria condition that bidder $t$ employs the best subsequent strategy in any subgame.
\end{proof}

\begin{lemma}
	\label{lem:single-bidder-lying-loss}
	Consider any day $t$ and the bidder's utility on that day.
	The utility of bidding some $b_t$ such that $|b_t - v_t| > 2 \alpha$ is worse than that of truthful bidding by at least $\frac{\alpha^3}{2}$. 
\end{lemma}

\begin{proof}
	Suppose $b_t < v_t$.
	Then, with probability $\frac{\alpha}{K}$ the algorithm price at $b_t + \alpha$ on day $t$. 
	Then, the utility of bid $b_t$ is zero while that of truthful bid is at least $\alpha$.
	Further, truthful bid is never worse than bid $b_t$ in the current round.
	So the total utility loss on day $t$ of bid $b_t$ is at least $\alpha \cdot \frac{\alpha}{K} > \frac{\alpha^3}{2}$.
\end{proof}

\begin{lemma} 
	\label{lem:single-bidder-deviation}
	On any day $t$, we have $b_t \in [v_t - 2 \alpha, v_t + 2 \alpha]$ for
	\begin{itemize}[topsep=1mm, parsep=0mm, itemsep=0mm]
		\item $\alpha = (4\tau \epsilon)^{1/3}$ under the assumption of large market; or
		\item $\alpha = (\frac{4\epsilon}{1-\gamma})^{1/3}$ under the assumption of impatient bidders. 
	\end{itemize}
\end{lemma}

\begin{proof}
	\emph{(Large market)~}  
	Suppose for contrary that $|b_t - v_t| > 2\alpha$.
	By Lemma~\ref{lem:single-bidder-lying-loss}, the utility in the current round is worse than that of bidding truthfully by at least $\frac{\alpha^3}{2}$.
	Further, Lemma~\ref{lem:single-bidder-stability}, the utility in subsequent round is better than that of truthful bidding by at most an $e^\epsilon < 1 + 2 \epsilon$ multiplicative factor plus an $\delta T$ additive factor. 
	Due to the large market assumption, the maximum utility in subsequent round is less than $\tau - 1$. 
	Further recall our choice of $\delta = \frac{\epsilon}{T}$.
	Hence, the above bound translates to having at most $2 \epsilon (\tau - 1) + \epsilon < 2 \epsilon \tau$ extra utility in subsequent rounds compared to truthful bidding.
	For $\alpha = (4\tau \epsilon)^{1/3}$, the loss in the current round is more than the gain in subsequent rounds, contradicting that $b_t$ is an equilibria bid.

	\emph{(Impatient bidders)~}
	The proof is almost identical, except that the maximum (discounted) utility in subsequent rounds is now bounded by $\gamma + \gamma^2 + \dots = \frac{\gamma}{1-\gamma} = \frac{1}{1-\gamma} - 1$ (instead of $\tau - 1$).
	%
\end{proof}

\begin{corollary} 
	\label{cor:single-bidder-gt-regret}
	The game-theoretic regret is bounded by $2 \alpha T$ for
	\begin{itemize}[topsep=1mm, parsep=0mm, itemsep=0mm]
		\item $\alpha = (4\tau \epsilon)^{1/3}$ under the assumption of large market; or
		\item $\alpha = (\frac{4\epsilon}{1-\gamma})^{1/3}$ under the assumption of impatient bidders. 
	\end{itemize}
\end{corollary}

\begin{proof}
	By Lemma~\ref{lem:single-bidder-deviation}, $|b_t - v_t| \le 2 \alpha$.
	Suppose $p^*$ is the optimal fixed price w.r.t.\ $\{v_t\}_t$. 
	Consider fixed price $p^* - 2\alpha$ w.r.t.\ $\{b_t\}_t$.
	Every time $p^*$ has a sale w.r.t.\ $\{v_t\}_t$, $p^* - 2\alpha$ also gets a sale w.r.t.\ $\{b_t\}_t$.
	So the revenue of price $p^* - 2\alpha$ w.r.t.\ $\{b_t\}_t$ is at least the optimal w.r.t.\ $\{v_t\}_t$ minus $2 \alpha T$.
\end{proof}

\subsection{Bounding regret} \label{subsection:full-one-fold-regret}
We prove the regret under the large market assumption. 
The case of impatient bidders is almost identical.
Putting together Corollary~\ref{cor:single-bidder-learning-regret-sketch} and Corollary~\ref{cor:single-bidder-gt-regret}, the regret of Algorithm~\ref{alg:single-bidder-sketch} is at most
%
$
\O \big( \sqrt{\log K} (\sigma \sqrt{\log T} + \frac{T}{\sigma \sqrt{\log T}}) + \alpha T \big)
$
for $\alpha = (4 \tau \epsilon)^{1/3}$.
This means that we shall set $\epsilon = \Theta( \frac{\alpha^3}{\tau} )$ and, thus,
%
$
\sigma 
= 
\tilde{\Theta} \big( \frac{\sqrt{K}}{\varepsilon} \big)
= 
\tilde{\Theta} \big( \tau \alpha^{-3.5} \big)
$.
So the 2nd term in the above regret bound is negligible compared to $\alpha T$.
The regret bound becomes
%
$
\tilde{\O} \big(  \tau \alpha^{-3.5} \big) + \O \big( \alpha T \big) \le \O \big( \alpha T \big)
$
%
if $T \ge \tau \alpha^{-4.5}$.

\subsection{Improved Algorithm}\label{subsection:single-autcion-2Dtree}

\paragraph{Two-fold Tree-aggregation}
We then give a two-fold tree-aggregation which reduces the scale of noise added to the cumulative gains from $\sqrt{K}$ to $\polylog K$ in terms of the dependence on $K$. 
We exploit the special structure of the gain vector $\vec{g}_t$ to further reduce this number to $\log K \log T$.

For convenience of notations in proofs, we consider prices in decreasing order in the analysis of two-fold tree-algorithm, i.e., $p_1 = 1 = (K-1)\alpha, \dots, p_{K-1}=\alpha, p_K=0$.
For any bid $b_t = p_{i_t} = (K - i_t) \alpha$, the gain vector is $\vec{g}_t = \big( 0, \dots, 0, (K - i_t) \alpha, (K - i_t - 1) \alpha, \dots, \alpha, 0 \big)$.
%
%
We denote the identification vector of $b_t$ as $\vec{i}(b_t)=(0,\dots,0,1,0,\dots,0)$, where only the $i_t$-th entry is $1$ and other entries are zero. 
Further denote as $\vec{c}(b_t)=(0,\dots,0,1,\dots,1)$ the prefix sum vector of $\vec{i}(b_t)$.
Then, the gain vector can be expressed as $\vec{g}_t = P \vec{c}(b_t)$,
%
%
where $P= \text{diag}(1, 1-\alpha, \dots, \alpha, 0)$.

In other words, the cumulative gains can be expressed as the prefix sum over both the time horizon and the expert horizon, i.e., $G_{ti} = p_i \cdot \sum_{t' = 1}^t \sum_{i' = 1}^{i} \vec{1}({b_{t'} = p_{i'}})$.
%
%
Hence, by using the tree-aggregation technique on both horizons, we can maintain $KT$ partial sums such that each $b_t$ contributes to at most $\log K \log T$ partial sums, and each $G_{ti}$ can be expressed as the sum of at most $\log K \log T$ partial sums.
Using the definitions of $\Gamma_i$ and $\Lambda_i$, we present the two-fold tree-aggregation in Algorithm~\ref{alg:tree-aggregation-2d}.
Then, our algorithm (Algorithm~\ref{alg:single-bidder-sketch}) uses the two-fold tree-aggregation to keep track of the cumulative gain.

\begin{algorithm}[t]
	\caption{~Two-fold Tree-aggregation}
	\label{alg:tree-aggregation-2d}
	\begin{algorithmic}[1]
		\STATE \textbf{input:~} dimension $K$, bid $b_t \in [0, 1]$ of each round $t$, noise scale $\sigma$.
		\STATE \textbf{internal states:~} noisy partial sums $A_{ti}$ for all $t \in [T]$ and all $i \in [K]$.
		\STATE \textbf{initialize:~} ${A}_{ti} = {\mu}_{ti}$ for all $t \in [T]$ and all $i\in[K]$, with $\mu_{ti}$'s i.i.d.\ from $N(0, \sigma^2)$.
		\FOR{$t = 1, 2, \dots, T$} 
		\STATE Receive bid $b_t = p_{i_t} = (K - i_t) \alpha$ as input.
		\STATE Let $A_{ji} = A_{ji} + 1$ for all $j$ s.t.\ $t \in \Lambda_j$ and for all $i$ s.t. $i_t \in \Lambda_i$. 
		\STATE \textbf{Output} $\tilde{\vec{G}}_t$ with $\tilde{G}_{ti} = p_i \big(\sum_{j \in \Gamma_t}\sum_{k \in \Gamma_i} A_{jk} + \nu_{ti}$\big), where $\nu_{ti}$'s are i.i.d.\ from normal distribution ${N} \big(0, \big( (\log K+1)(\log T + 1) - |\Gamma_t||\Gamma_i| \big) \sigma^2 \big)$. 
		\ENDFOR
	\end{algorithmic}
\end{algorithm}

\begin{lemma}
	\label{lem:tree-aggregation-privacy-previous-2d}	
	Fixed any $t_0 \in [T]$, the values of internal states $A_{ti}$'s after time $t_0$ are $(\epsilon, \delta = \frac{\epsilon}{T})$-differentially private with noise scale
	$\sigma \ge \frac{8\log T \log K}{\varepsilon}  \sqrt{\ln(\frac{ \log K \log T}{\delta})}$.
\end{lemma}
\begin{proof} 
	Note that each bid $b_t$ contributes at at most $\log K \log T$ of the partial sums with sensitivity $1$. 
	The lemma then follows by the privacy guarantee of Gaussian mechanism (Lemma~\ref{lemma:gaussian-mechanism}) and the composition theorem (Lemma~\ref{lemma:composition-theorem}).
\end{proof}

\subsection{Bounding learning regret (improved algorithm)}
We then analyze the learning regret of Algorithm \ref{alg:single-bidder-sketch} with two-fold tree-aggregation.
For any day $t$, we denote the error term in the cumulative gain up to time $t$, i.e., $\tilde{\vec{G}}_t - \vec{G}_t$, as $\vec{u}_t$. 
Note that $u_{ti}$ follows a normal distribution ${N}(0,\sigma^2p_i^2(\log K+1)(\log T+1))$, where $p_i$ is the price of the $i$-th expert. 
Hence, $\vec{u}_t$ follows a Gaussian distribution with correlations among entries due to tree-aggregation on the expert/price horizon. 
Further, all $\vec{u}_t$'s follow the same Gaussian, but are correlated over the time horizon.
Due to the correlations over the expert horizon, we cannot directly apply the results of follow-the-perturbed-leader with independent Gaussian noise (e.g., \cite{abernethy2014online}).

Instead, we analyze the regret using the more general framework of Gradient-based-prediction-algorithm (GBPA) by~\citet{abernethy2014online}. 
In each round $t$, the seller picks a suitable function $\Phi_t$, and then selects an expert/strategy from distribution $\nabla \Phi_t(\vec{G}_{t-1})$ (given that it is a distribution). 
Algorithm \ref{alg:single-bidder-sketch} effectively use a fixed (correlated) Gaussian smoothing potential function $\Phi_t=\Phi$ with $\Phi(\vec{G}_{t-1})=\mathbb{E}_{\vec{u}_{t-1}} \left[ \max_{i}\left\{{G}_{(t-1)i} + u_{(t-1)i}\right\} \right] $.


\begin{lemma}\label{lemma:gbpa}
	(Lemma 2 of \citet{abernethy2014online}) 
	The learning regret w.r.t.\ the best fixed discretized price is at most
	\begin{equation}\label{equ:gbpa}
		\textstyle
		\max_{i \in [K]} G_{Ti} - \Phi(\vec{G}_T) + \Phi(0) + \sum_{t=1}^T  D_{\Phi}(\vec{G}_{t},\vec{G}_{t-1})
	\end{equation}
	where the last term is the Bregman divergence defined as $D_{\Phi}(\vec{y},\vec{x})=\Phi(\vec{y})-\Phi(\vec{x})-\langle \nabla\Phi(\vec{x}),\vec{y}-\vec{x}\rangle$.
\end{lemma}

Note that Bregman divergence measures how fast the gradient changes, which can be bounded exploiting differential privacy of the algorithm. 

\begin{lemma}\label{lemma:bregman}
	For any $t \in [T]$, we have $D_{\Phi}(\vec{G}_t,\vec{G}_{t-1})\leq 2 \varepsilon + 2 \delta$.
\end{lemma}
\begin{proof}
	From the definition of Bregman divergence, we have
	\begin{align*}
		D_{\Phi}(\vec{G}_t,\vec{G}_{t-1}) &= \Phi(\vec{G}_t) - \Phi(\vec{G}_{t-1}) -\langle \nabla \Phi(\vec{G}_{t-1}), \vec{G}_t-\vec{G}_{t-1} \rangle \\
		&\leq |\langle \vec{G}_t-\vec{G}_{t-1},\nabla\Phi(\vec{G}_t)-\nabla\Phi(\vec{G}_{t-1})\rangle| \\
		&=|\langle \vec{g}_t, \nabla\Phi(\vec{G}_t)-\nabla\Phi(\vec{G}_{t-1})\rangle|\\
		&\leq ||\vec{g}_t||_{\infty} ||\nabla\Phi(\vec{G}_t)-\nabla\Phi(\vec{G}_{t-1})||_{1}\\
		&\leq ||\nabla\Phi(\vec{G}_t)-\nabla\Phi(\vec{G}_{t-1})||_{1}
	\end{align*}
	where the first inequality is by convexity, the second inequality is by H\"{o}lder inequality, and the last inequality is because of $\lVert \vec{g}_t \rVert_{\infty} \le 1$.
	
	Let $\vec{q}_t = \nabla \Phi(\vec{G}_{t})$ and $\vec{q}_{t-1} = \nabla \Phi(\vec{G}_{t-1})$ be the pricing probabilities on day $t+1$ and day $t$ respectively. 
	For any $i \in [K]$, let $S_i=\{\vec{y}\in \mathbb{R}^K: i=\argmax_j{y}_j\}$.
	Then, by definition, we have
	\[
	q_{ti}= \pr \big[ \vec{G}_{t}+\vec{u}_{t}\in S_i \big] \quad\text{and}\quad 	
	q_{(t-1)i}= \pr \big[ \vec{G}_{t-1}+\vec{u}_{t-1}\in S_i \big] ~.
	\]
	Note that $\vec{u}_{t}$ and $\vec{u}_{t-1}$ are identically distributed (yet not independent).
	Further, the noise is sufficient for guaranteeing $(\varepsilon,\delta)$-differentially private.
	For any subset $L \subseteq [K]$, considering $\vec{G}_t$ and $\vec{G}_{t-1}$ as two neighboring datasets, we have
	\[ 
	\textstyle
	\sum_{i \in L} q_{ti} = \Pr \left[ \vec{G}_t+ \vec{u} \in \bigcup_{i \in L} S_i \right] \leq e^{\varepsilon} \Pr \left[ \vec{G}_{t-1}+ \vec{u}\in \bigcup_{i \in L}S_i \right] +\delta = e^{\varepsilon} \cdot \sum_{i \in L} q_{(t-1)i} + \delta ~,
	\]
	and similarly $\sum_{i \in L} q_{(t-1)i} \le e^{\varepsilon} \cdot \sum_{i \in L} q_{ti} + \delta$.
	Therefore, we have
	\begin{align*}
		\textstyle
		||\vec{q}_t-\vec{q}_{t-1}||_1 
		&   
		=
		\sum_{i : q_{ti} \geq q_{(t-1)i}} \big( q_{ti} - q_{(t-1)i} \big) 
		+
		\sum_{i : q_{ti} < q_{(t-1)i}} \big( q_{(t-1)i} - q_{ti} \big) \\
		&
		\leq 
		(e^\varepsilon-1) \cdot \sum_{i : q_{ti} \geq q_{(t-1)i}} q_{(t-1)i} + \delta 
		+
		(1 - e^{-\varepsilon}) \cdot \sum_{i : q_{ti} < q_{(t-1)i}} q_{(t-1)i} + \delta \\
		&
		\le (e^\varepsilon-1) \sum_i q_{(t-1)i} + 2 \delta \\
		&
		= (e^\varepsilon-1) + 2 \delta \le 2 \varepsilon + 2 \delta ~.
	\end{align*}
\end{proof}

\begin{lemma}\label{lem:twofold-tree-regret}
	The learning regret of Algorithm~\ref{alg:single-bidder-sketch} is at most $\O \big( \varepsilon T + \alpha T \big) + \tilde{\O} \big( \sigma \big)$.
\end{lemma}

\begin{proof}
	We analyze the terms in Eq.~\eqref{equ:gbpa} of Lemma~\ref{lemma:gbpa}.
	First, consider $\Phi(0) = \mathbb{E}_{\vec{u}} [\max_iu_i]$.
	%
	%
	Recall that $u_i\sim{N}(0,\sigma^2p_i^2\log K\log T)$.  

	For some parameter $a > 0$ to be determined later, we have
	\begin{align*}
		\textstyle
		\exp \big( a \cdot \E[ \max_i u_i] \big) & \leq \E \big[ \exp(a \cdot \max_i u_i) \big] \leq \sum_i \E \big[ \exp(a \cdot u_i) \big] \\
		& =\sum_{i} \exp \big( a^2p_i^2\sigma^2\log K\log T/2 \big) \\
		& \leq K \cdot \exp(a^2 \sigma^2\log K\log T/2) ~,
	\end{align*}
	the equality follows by the moment generating function of Gaussian random variables. 
	Hence,
	\[
	\textstyle
	\E \big[ \max_i u_i \big] \leq \frac{a\sigma^2\log K\log T}{2}+\frac{\ln K}{a},\forall a ~.
	\]
	Setting $a=\frac{\sqrt{2}}{\sigma\sqrt{\log T}}$, we get that 
	$\Phi(0) = \E[\max_i u_i] \le \sigma\log K\sqrt{2\log T}= \tilde{\O}(\sigma)$.
	
	Next we bound $\max_i G_{Ti} - \Phi(\vec{G}_T)$. Note that $\max_i G_{Ti}$ is a convex function of $\vec{G}_T$, we have 
	\[
	\textstyle
	\Phi(\vec{G}_T)=\E_{\mathbf{u}_T}[\max_i\{G_{Ti}+u_{Ti}\}]\ge \max_i\{{G}_{Ti}+\E_{{u}_T}[u_{Ti}]\}=\max_i G_{Ti}.
	\]
	Further, Lemma~\ref{lemma:bregman} bounds each term in the summation.
	Putting together with our choice of $\delta = \frac{\epsilon}{T}$, the regret w.r.t.\ the best fixed discretized price is at most $\O (\varepsilon T) + \tilde{\O} ( \sigma)$.
	%
	%
	Finally, the regret w.r.t.\ discretized prices differs from the learning regret by at most $\alpha T$. 
	So the lemma follows.
\end{proof}


\subsection{Bounding regret (Proof of Theorem~\ref{thm:single-bidder})}\label{subsection:full-two-fold-regret}
The stability of future utility (Lemma~\ref{lem:single-bidder-stability}) and the bounds on the game-theoretic regret (Corollary~\ref{cor:single-bidder-gt-regret}) still hold, because the internal states' differential privacy with Lemma~\ref{lem:tree-aggregation-privacy-previous-2d} in place of Lemma~\ref{lem:tree-aggregation-privacy}.

Combine with the upper bound on learning regret in Lemma~\ref{lem:twofold-tree-regret}, the total regret is at most $\O \big(\epsilon T+\alpha T \big) + \tilde{\O} (\sigma)$,
with $\epsilon=\Theta(\frac{\alpha^3}{\tau})$ under large market assumption.
Further recall that $\sigma = \tilde{\O}(\frac{1}{\epsilon})$, the regret is at most $\O(\alpha T)$ if $T \ge \tilde{\O} (\tau\alpha^{-4})$.
The case of impatient bidders is similar.

\subsection{Online Single-buyer Posted Pricing (Bandit Setting)}
\label{subsection:bandit_sketch}
In posted pricing, the seller cannot access the exact bid value of the buyer, and thus cannot obtain the gain vector $\mathbf{g}_t$.
This results in the failure of normal FTPL. 
Here we follow a widely-used method in online learning when the learner only has partial information. In each round we calculate an unbiased estimated gain vector $\bar{\mathbf{g}}_t$, i.e. $\E[\bar{\mathbf{g}}_t]=\mathbf{g}_t$. The expectation is taken w.r.t. the price chosen in round $t$. 
Formally, we propose our algorithm in bandit setting as following.

\begin{algorithm}[H]
	\caption{Extended FTPL in bandit setting} \label{Alg:FTPL-bandit-setting}
	\begin{algorithmic}[1]
		\STATE \textbf{input:} regret parameter $\alpha$, $K=\frac{1}{\alpha}+1$, privacy parameter $\varepsilon,\delta=\frac{\varepsilon}{T}$
		\STATE \textbf{initialize:} tree-aggregation (Alg.~\ref{alg:tree-aggregation}) with noise scale $\sigma = \frac{8K}{\alpha\varepsilon} \log T \sqrt{\ln\frac{\log T}{\delta}}$.
		\FOR {$t=1,\ldots,T$}
		\STATE Let $\tilde{\mathbf{q}}_{t-1}=(1-\alpha)\mathbf{q}_{t-1}+\frac{\alpha}{K} \textbf{1}$, where ${q}_{(t-1)i}=\Pr[i=\argmax_{i'} \{\tilde{G}_{(t-1){i'}}\}]$ for each $i \in [K]$.
		
		\STATE Select ${i_t}$ from $\tilde{\mathbf{q}}_{t-1}$, observe the reward $g_{ti_t}$, i.e. the payment.
		\STATE Calculate $\bar{\mathbf{g}}_t$ 
		\begin{displaymath}
		\bar{{g}}_{ti} = \left\{ \begin{array}{ll}
		 g_{ti_t}/\tilde{q}_{(t-1)i} &\textrm{if $i=i_t$,} \\
		0 & \textrm{otherwise.}
		\end{array} \right. 
		\end{displaymath}
		\STATE Let $\bar{\mathbf{G}}_t=\bar{\mathbf{G}}_{t-1}+\bar{\mathbf{g}}_t$.
		\STATE Update tree-aggregation with $\bar{\vec{g}}_t$ and get $\tilde{\vec{G}}_t$.
		\ENDFOR
	\end{algorithmic}
\end{algorithm}

Note that $\bar{\mathbf{g}}_t$ doesn't have a special structure that can reduce the noise scale. We are using normal tree aggregation in this algorithm.

Follow the same routine in full information setting, we first bound the learning regret. Then we prove that Algorithm~\ref{Alg:FTPL-bandit-setting} also has \textit{Stability of Future Utility} (Lemma~\ref{lem:single-bidder-stability}), therefore we can bound the game-theoretic regret. Finally we have the formal theorem in bandit setting:
\begin{theorem}
\label{thm:bandit_total_regret}
	
	For any $\alpha>0$, there is an online algorithm in bandit setting with regret $\O(\alpha T)$ when $T\ge \tilde{\O}(\tau\alpha^{-4})$ under the large market assumption, or $T\ge \tilde{\O}(\frac{\alpha^{-4}}{1-\gamma})$ under the impatient bidder assumption.
\end{theorem}

The full proof is shown in the Appendix.

\section{Multi-bidder Case}
\label{sec:multi-sketch}
We take the same online learning formulation as in the single-bidder case, treating each discretized price that is a multiple of $\alpha$ between $0$ and $1$ as an expert.
An expert $j$'s gain on any day is the revenue of Vickrey auction with reserve price $(j-1)\alpha$ w.r.t.\ the bids on that day.
Note that the gain now could be as large as $m$ since the seller has $m$ copies for sale.
we further normalize the gain by dividing it by $m$.
Given bids $\vec{b}_t$, suppose for any $j$ there are $m_j$ of the $b_{ti}$'s that are at least $(j-1)\alpha$, and the $(m+1)$-th highest bid is $(j' - 1)\alpha$, the gain vector $\vec{g}_t$ is then defined as:
\[
g_{tj} = \begin{cases}
(j-1) \alpha m_j & \text{if $m_j \le m$;}\\
(j' - 1) \alpha m& \text{otherwise.}
\end{cases}
\]
Then, $g_{tj}$ equals the revenue of running Vickrey with reserve price $(j-1)\alpha$ w.r.t.\ bids $\vec{b}_t$.

\begin{theorem}
	\label{thm:multi-bidder}
	For any $\alpha > 0$, our algorithm runs an approximate version of Vickrey with an anonymous reserve price on each day with regret $\le \alpha m T$ against the best fixed reserve price if:
	\begin{enumerate}[topsep=1mm, parsep=0mm, itemsep=0mm]
		\item $T \ge \tilde{O} \big( \frac{\tau n}{m \alpha^{4.5}} \big)$, and $m \ge \tilde{O}( \frac{\sqrt{\tau n}}{\alpha^3})$ given large market; or
		\item $T \ge \tilde{O} \big( \frac{n}{(1 - \gamma) m \alpha^{4.5}} \big)$, and $m \ge \tilde{O}( \frac{\sqrt{n}}{\sqrt{1 - \gamma} \alpha^3})$ given impatient bidders.
	\end{enumerate}
\end{theorem}

%
\subsection{Algorithm}
With some small probability we randomly pick a subset of bidders and offer each of them a copy of the good with a random price to ensure lying is costly in the current round. 
We pick the reserve price on each day using follow-the-perturbed-leader implemented with tree-aggregation. 
Simply running Vickrey with the chosen reserve price does not guarantee stability of future utility, however, because a bidder's current bid can now affect other bidders' subsequent bids through the allocations and payments in the current round.
Instead, we use an algorithm of \citet{hsu2016private} to get a set $S$ and a price $p$ that are approximations of the set of top-$m$ bidders and Vickrey price.

\begin{algorithm}[H]
	\caption{~Online Pricing (Multi-bidder Case)}
	\begin{algorithmic}[1]\label{alg:multi-full}
		\STATE \textbf{input:} regret parameter $\alpha$, $K=\frac{1}{\alpha}+1$, privacy parameter $\epsilon$, $\delta=\frac{\epsilon}{T}$, $E = \tilde{\O} \big( \frac{1}{\alpha^2 \epsilon} \big)$.
		\STATE \textbf{initialize} tree-aggregation with noise scale $\sigma = \frac{8 \sqrt{K}\log T}{\varepsilon}  \sqrt{ \ln \frac{\log T}{\delta}}$.	
		\FOR {$t=1,\ldots,T$} 
		\STATE With probability $\alpha$, pick a subset $S \subseteq [n]$ of size $m$ and $j \in [K]$ uniformly at random. 
		\STATE Otherwise:
		\STATE\hspace{\algorithmicindent} Pick $j_1$ that maximizes $\tilde{\vec{G}}_{(t-1)j}$ from tree-aggregation.\label{step:multi-pick-price}
		\STATE\hspace{\algorithmicindent} Run PMatch($\alpha, \rho = \alpha, \epsilon$) \cite{hsu2016private} to get a set $S$ of $\le m - E$ bidders and a price $p = j_2 \alpha$. \label{step:pmatch}
		\STATE\hspace{\algorithmicindent} Let $j = \max \{ j_1 - 1, j_2 \}$.
		\STATE Offer a copy of good to each $i \in S$ at price $(j - 1)\alpha$.
		\STATE Observe bid vector $\vec{b}_t$; update tree-aggregation with the normalized gain vector $\frac{1}{m} \vec{g}_t$. \label{step:multi-update}
		\ENDFOR		
	\end{algorithmic}
\end{algorithm}

\subsection{Bounding Learning Regret}


We decompose the learning regret into four parts:
\begin{enumerate}
	\item The difference between the regret w.r.t.\ the best fixed discretized price and the actual regret;
	\item The regret due to picking a random price price with probability $\alpha$ (step 4);
	\item The regret w.r.t.\ the best fixed discretized price if we omit step 4, and replace steps 7-9 in Algorithm~\ref{alg:multi-full} with a Vickrey auction with reserve price $(j_1 - 1) \alpha$; and
	\item The difference in revenue between running steps 7-9 in Algorithm~\ref{alg:multi-full} and running a Vickrey auction with reserve price $(j_1 - 1) \alpha$.
\end{enumerate}
The first part is bounded by $\alpha m T$ since rounding down any fixed price the closest multiple of $\alpha$ loses at most $\alpha m$ in revenue per round.
The second part is also bounded by $\alpha m T$ since the maximum gain is $m$ per round.
It remains to bound the third part using the regret analysis of follow-the-perturbed-leader, and the fourth part of PMatch allocation.

\begin{lemma}
	The regret w.r.t.\ the best fixed discretized price if we omit step 4, and replace steps 7-9 in Algorithm~\ref{alg:multi-full} with a Vickrey auction with reserve price $(j_1 - 1) \alpha$ is at most 
	\[
	\textstyle
	\O \left( m\sqrt{\log K} \big( \sigma \sqrt{\log T} + \frac{T}{\sigma \sqrt{\log T}} \big) \right) ~.
	\]
\end{lemma}
\begin{proof}
	Note that the amount of noise we add to each coordinates of $\vec{G}_t$ follows $N(0, (\log T + 1) \sigma^2)$.
	By Theorem 8 of \citet{abernethy2014online}, we get that the regret in terms of the normalized gain is bounded by $\O \left( \sqrt{\log K} \big( \sigma \sqrt{\log T} + \frac{T}{\sigma \sqrt{\log T}} \big) \right)$.
	%
	%
	Multiplying it by $m$ proves the lemma.
	%
	%
\end{proof}
\begin{lemma}[\citet{hsu2016private}]
	\label{lem:pmatch}
	The set of bidders $S$ and the price $p$ satisfy:
	\begin{enumerate}[topsep=1mm, parsep=0mm, itemsep=0mm]
		\item $S$ is $(\epsilon, \delta)$-jointly differentially private;
		\item $p$ is $(\epsilon, \delta)$-differentially private;
		\item $m - 2E \le |S| \le m - E$;
		\item all bidders in $S$ have values at least $p - \alpha$;
		\item at most $E$ bidders outside $S$ have values at least $p$. 
	\end{enumerate}
\end{lemma}

\begin{lemma}
	\label{lem:multi-bidder-regret-sketch}
	For any $j^* \in [K]$, the revenue of running Vickrey with reserve $p^* = (j^* - 1) \alpha$ is no more than that of running steps 7-9 in Algorithm~\ref{alg:multi-full} with $j_1 = j^*$ plus $\O \big( E + \alpha m \big)$.
\end{lemma}

\begin{proof}
	Suppose $p' = j' \alpha$ is the $(m+1)$-th highest bid.
	The winners in Vickrey pays $\max \{ p^*, p' \}$.
	With $j_1 = j^*$, $(j_1 - 2) \alpha = p^* - \alpha$.
	Claims 3 and 5 of Lemma~\ref{lem:pmatch} imply $p \ge p'$ and, thus, $(j_2 - 1) \alpha \ge p' - \alpha$.
	Hence, the price offered in step 9 is at least $\max \{ p^*, p' \} - \alpha$.
	
	It remains to show the number of sales by steps 7-9 is less than that of Vickrey by at most $\O(E)$.
	If $(j_1 - 2) \alpha$ is offered in step 9, then the number of sales by the algorithm is at least that of Vickrey with reserver $p^*$ minus $E$ due to claim 5 of Lemma~\ref{lem:pmatch}.
	If $(j_2 - 1) \alpha = p - \alpha$ is offered in step 9, then the number of sales is at least $m - 2E$ due to claim 3 and 4 of Lemma~\ref{lem:pmatch}.
	Hence, it is less than the number of sales of Vickrey by at most $2E$.
	In both cases, the lemma follows.
\end{proof}

Putting together we have the following bound on the learning regret.

\begin{lemma}
	\label{lem:multi-bidder-learning-regret}
	The learning regret of Algorithm~\ref{alg:multi-full} is at most
	\[
	\textstyle
	\O \left( m\sqrt{\log K} \big( \sigma \sqrt{\log T} + \frac{T}{\sigma \sqrt{\log T}} \big) + \alpha m T + E T \right) ~.
	\]
\end{lemma}

\subsection{Bounding Game-theoretic Regret}
We first establish the stability of future utility.

\begin{lemma}[Stability of Future Utility, Multi-bidder]
	\label{lem:multi-bidder-stability}
	For any bidder and any day $t$ on which he comes, the bidder's equilibria utilities in subsequent rounds in the subgames induced by different bids on day $t$ differ by at most an $e^{3\epsilon}$ multiplicative factor plus a $3\delta T$ additive factor.
\end{lemma}

\begin{proof}
	We shall abuse notation and refer to the bidder of concern as bidder $t$ even though there are other $n-1$ bidders who also come on day $t$.
	Fixed any bidder $t$'s strategy for subsequent days (after day $t$).
	That is, fixed the (randomized) bidding function on any subsequent day $t'$ as a function only on his bids and auction outcomes between day $t$ and $t'$ (exclusive).
	We shall interpret any other bidder's strategy on any subsequent day $t'$ as a (randomized) function that depends on his bids and auction outcomes between day $t$ (inclusive) and $t'$ (exclusive).
	By allowing the other bidders' strategies to depend on what they observe on day $t$, we can treat them as fixed regardless of bidder $t$'s bid on day $t$, which cannot be observed by the other bidders.
	
	Let us consider the resulting utilities for bidder $t$ in the subgames induced by two distinct bids on day $t$ given any fixed subsequent strategy of bidder $t$.
	%
	The execution of the online pricing algorithm, i.e., the algorithm together with the bidders' strategies in subsequent rounds, after round $t$ is a post-processing on (1) the internal states of the tree-aggregation algorithm after day $t$, (2) the price posted on day $t$, and (3) the allocations to bidders other than $t$.
	Each of these parts is $(\epsilon, \delta)$-differentially private w.r.t.\ bidder $t$'s bid on day $t$ due to Lemma~\ref{lem:tree-aggregation-privacy} and claim 1 and 2 of Lemma~\ref{lem:pmatch}.
	Therefore, the utilities of any fixed subsequent strategy of bidder $t$ in the two subgames differ by at most an $e^{3\epsilon}$ multiplicative factor plus a $3 \delta T$ additive factor.
	The lemma then follows by the equilibria condition that bidder $t$ employs the best subsequent strategy in any subgame.
\end{proof}

Next, we will lower bound of cost of lying. 
Recall that our single-bidder algorithm always uses truthful single-round auctions.
In contrast, steps 7-9 of Algorithm~\ref{alg:multi-full} is not truthful in general and, thus, a bidder may be able to gain even in the current round by lying.
We resort to a weaker argument that lower bound the loss only for underbidding, which is sufficient for our regret analysis.
We will use the next lemma that follows by the deferred-acceptance nature of PMatch.

\begin{lemma}[\citet{hsu2016private}, implicit]
	\label{lem:pmatch-underbid}
	Fixed any random bits of PMatch, a bidder cannot change the set $S$ of bidders and price $p$ by underbidding without excluding himself from $S$.
\end{lemma}

\begin{lemma}
	\label{lem:multi-bidder-lying-loss}
	Consider any day $t$ and the bidder $i$'s utility on that day.
	The utility of bidding at some $b_{ti}$ such that $|b_{ti} - v_{ti}| > 2 \alpha$ is worse than that of truthful bidding by at least $\frac{\alpha^3m}{2n}$. 
\end{lemma}

\begin{proof}
	With probability $\alpha$, a bidder $i$ would be chosen into the candidate allocation group $S$ with probability at least $\frac{\alpha m}{n}$ and be offered a randomly chosen price. 
	Suppose $b_{ti} < v_{ti}$, the utility of bid $b_{ti}$ is zero while that of truthful bidding is at least $\alpha$.
	Further, truthful bidding is never worse than under bidding $b_{ti}$ in the current round due to Lemma~\ref{lem:pmatch-underbid}.
	So the total utility loss on day $t$ of bidding $b_{ti}$ is at least $\alpha \cdot \frac{\alpha m}{Kn} > \frac{\alpha^3 m}{2n}$.
\end{proof}

Then, we can control the level of underbidding with the next lemma.

\begin{lemma} 
	\label{lem:multi-bidder-deviation}
	On any day $t$ for any bidder $i$, we have $b_{ti} \in [v_{ti} - 2\alpha, v_{ti} + 2\alpha]$ for
	\begin{itemize}[topsep=1mm, parsep=0mm, itemsep=0mm]
		\item $\alpha = (\frac{12 \tau \epsilon n}{m})^{1/3}$ under the assumption of large market; or
		\item $\alpha = (\frac{12 \epsilon n}{(1-\gamma)m})^{1/3}$ under the assumption of impatient bidders. 
	\end{itemize}
\end{lemma}

\begin{proof}
	Again, we prove it under large market and the other case is almost identicy.
	Suppose for contrary that $|b_{ti} < v_{ti}| > 2 \alpha$ for some bidder $i$ on day $t$.
	By Lemma~\ref{lem:multi-bidder-lying-loss}, the loss on day $t$ is at least $\frac{\alpha^3 m}{2n}$.
	By Lemma~\ref{lem:multi-bidder-stability}, our choice of $\delta = \frac{\varepsilon}{T}$, and that the total future utility of the bidder is at most $\tau - 1$ under large market, the gain in future utility is upper bounded by 
	\[
	(e^{3\epsilon} - 1) (\tau - 1) + 3 \delta T < 6 \epsilon (\tau - 1) + 3 \epsilon < 6 \tau \epsilon ~.
	\]
	Comparing the loss in current round and the future gain gives a contradiction to the equilibria condition that the bidder plays best strategy in any subgame.
\end{proof}

Finally, we have the following upper bounds on game-theoretic regret as a corollary.

\begin{corollary} 
	\label{cor:multi-bidder-gt-regret}
	The game-theoretic regret is bounded by $2 \alpha mT$ for
	\begin{itemize}[topsep=1mm, parsep=0mm, itemsep=0mm]
		\item $\alpha = (\frac{12\tau \epsilon n}{m})^{1/3}$ under the assumption of large market; or
		\item $\alpha = (\frac{12\epsilon n}{(1-\gamma)m})^{1/3}$ under the assumption of impatient bidders. 
	\end{itemize}
\end{corollary}
\begin{proof}
	By Lemma~\ref{lem:multi-bidder-deviation}, for each bidder $i$, we have $|b_{ti}-v_{ti}| \le 2\alpha$. Since there are $m$ copy of goods, the game-theoretic regret becomes $2 \alpha mT$, similar argument to Corollary~\ref{cor:single-bidder-gt-regret}.
\end{proof}

\subsection{Bounding Regret (Proof of Theorem~\ref{thm:multi-bidder})}

Again, we prove it only for large market since the other case is almost identical.
Putting together the bounds for learning regret (Lemma~\ref{lem:multi-bidder-learning-regret}) and game-theoretic regret (Corollary~\ref{cor:multi-bidder-gt-regret}), we get that the regret of Algorithm~\ref{alg:multi-full} is at most
\[
\textstyle
\O \left( m\sqrt{\log K} \left( \sigma \sqrt{\log T} + \frac{T}{\sigma \sqrt{\log T}} \right) + \alpha m T + E T \right) ~,
\]
for $\alpha = (\frac{12 \tau \epsilon n}{m})^{1/3}$ under the large market assumption.

Suppose we further have $m \ge \frac{E}{\alpha} = \tilde{\O} \left( \frac{1}{\alpha^3 \epsilon} \right) = \tilde{\O} \left( \frac{\tau n}{\alpha^6 m} \right)$,
%
%
which is equivalent to the condition in Theorem~\ref{thm:multi-bidder} with $m \ge \tilde{\O} \left( \frac{\sqrt{\tau n}}{\alpha^3} \right)$.
%
%
Then, the regret bound simplifies to 
\[
\textstyle
\O \left( m\sqrt{\log K} \left( \sigma \sqrt{\log T} + \frac{T}{\sigma \sqrt{\log T}} \right) + \alpha m T \right) ~.
\]
Further, our choice of noise scale is $\sigma = \frac{8 \sqrt{K } \log T}{\varepsilon} \sqrt{\ln \frac{\log T}{\delta}} = \tilde{\O} \left( \frac{\tau n}{m \alpha^{3.5}}\right)$ (recall $K = \frac{1}{\alpha} + 1$).
%
%
So the regret bound becomes $\O \left( \frac{\tau n}{\alpha^{3.5}} \right) + O \left( \alpha m T \right)$.
%
%
Hence, the regret is at most $\O(\alpha mT)$ if $T \ge \tilde{\O} ( \frac{\tau m}{\alpha^{4.5}})$.



\bibliographystyle{plainnat}
\bibliography{matching}

\appendix

\section{Online Single-bidder Posted Pricing (Bandit Setting)}\label{sec:single-bandit}
In this section, we give an online learning algorithm for single-bidder in bandit setting. It is another well studied setting where the seller only knows whether the good is sold or not and the payment, instead of the exact bid value. The difficulty is that the seller cannot get the full information of gain vector $\mathbf{g}_t$.

Recall that we propose Algorithm \ref{Alg:FTPL-bandit-setting} in the main text. We first show some properties
in Subsection~\ref{subsection:single-bandit-alg}. Then we analyze the learning regret and game-theoretic regret of this algorithm in Subsection~\ref{subsection:single-bandit-regret} and \ref{subsection:single-bandit-game-regret}. Finally, we show a total regret of $\O(T\alpha)$, for $T\ge \tilde{\O}(\tau\alpha^{-4})$ under large market assumption, and for $T\ge \tilde{\O} (\frac{\alpha^{-4}}{1-\gamma})$ under impatient bidders assumption in Subsection~\ref{subsection:single-bandit-total-regret}.

\subsection{Private Algorithm for Posted Pricing}\label{subsection:single-bandit-alg}
We extended FTPL to bandit setting. Since the seller cannot receive the gain vector $\mathbf{g}_t$, 
we use an unbiased estimated gain vector $\bar{\mathbf{g}}_t$ with $\E[\bar{\mathbf{g}}_t]=\mathbf{g}_t$. The expectation is taken w.r.t. the randomness of algorithm in round $t$.
We also set the exploration to be a uniform distribution with probability $\alpha$. 
Recall that we define $\vec{q}_{t-1}$ as a distribution of price in day $t$. 

Similar to the full information setting, we have the following lemmas used in bandit setting.

\begin{lemma}
	\label{lem:tree-aggregation-noise-bandit}
	$\tilde{G}_{tj} - \bar{G}_{tj}$ follows the Gaussian distribution ${N}(0, \log T  \sigma^2)$ $\forall t \in [T]$ and $\forall j \in [K]$.
\end{lemma}


\begin{lemma}
	\label{lem:tree-aggregation-privacy-bandit}
	Fixed any $t_0 \in [T]$, the values of $\vec{A}_t$'s after time $t_0$ are $(\epsilon, \delta)$-differentially private for any $t_0$-neighboring datasets (differ only in the $t_0$-th entry), with $\sigma = \frac{8K}{\alpha\varepsilon} \log T \sqrt{\ln\frac{\log T}{\delta}}$.
\end{lemma}
\begin{proof}
	Let $D$ and $D'$ be $t_0$-neighboring datasets.
	The distribution of expert chosen in round $t_0$ is the same w.r.t. $D$ and $D'$. 
	Fix the values of $\vec{A}_t$'s after time $t_0-1$, applying Lemma \ref{lemma:gaussian-mechanism} and we obtain that the value of each $\vec{A}_t$ after $t_0$ is $(\frac{\epsilon}{\log T},\frac{\delta}{\log T})$-differentially private. 
	The $t_0$-th data in $D$ affects at most $\log T$ $\vec{A}_t$'s.
	By composition theorem (Lemma~\ref{lemma:composition-theorem}), the set of values of $\vec{A}_t$'s after time $t_0$ is $(\varepsilon,\delta)$-differentially private.
\end{proof}

\subsection{Bounding learning regret}\label{subsection:single-bandit-regret}
Denote the algorithm output as $\widetilde{\ALG}=\sum_{t\in [T]} {g}_{ti_t}=\sum_{t\in [T]} \langle\tilde{\mathbf{q}}_{t-1},\bar{\mathbf{g}}_t\rangle$. In expectation over the randomness of the algorithm, we have $G_{Ti}=\E[\bar{G}_{Ti}]$.
By the definition of $\tilde{\mathbf{q}}_t$, we get
$\E[\widetilde{\ALG}] \geq (1-\alpha)\sum_{t\in [T]} \E\langle {\mathbf{q}}_{t-1},\bar{\mathbf{g}}_t\rangle$.
Then
\begin{equation}
\textstyle
\label{eq:bandit-learning-regret-framework}
 \forall i\in[K], G_{Ti}-\E[\widetilde{\ALG}]\leq \E[ \bar{G}_{Ti}-\sum_{t\in[T]}\langle \mathbf{q}_{t-1},\bar{\mathbf{g}}_t\rangle]+\frac{\alpha}{1-\alpha}\E[\widetilde{\ALG}] ~.
\end{equation}
The last term is bounded by $2\alpha T$ because $\E[\widetilde{\ALG}]\leq T$ and $\alpha \le \frac{1}{2}$. By definition, the learning regret of Alg.~\ref{Alg:FTPL-bandit-setting} is $\OPT-\E[\widetilde{\ALG}]=\max_i\{G_{Ti}-\E[\widetilde{\ALG}]\}$, and thus can be upper bounded by RHS of Eq.(\ref{eq:bandit-learning-regret-framework}).

The proof is almost identical as in full information setting, using gain vector $\bar{\vec{g}}_t$ instead.
From Lemma~\ref{lem:tree-aggregation-noise-bandit}, we know that  $\tilde{\mathbf{G}}_t=\bar{\mathbf{G}}_t+\mathbf{u}_t$ with $\mathbf{u}_t \sim {N}(0,\sigma^2 \log^2 T\mathbf{I})$. Denote this normal distribution as $\mathcal{D}'$. To use the GBPA framework as in Section \ref{sec:single-autcion-regret-2d}, we choose 
$\Phi(\bar{\mathbf{G}})=\E_{\mathbf{u}\sim \mathcal{D}'}[\max_i\{\bar{G}_i+u_i\}]$.
Then $\bigtriangledown\Phi(\bar{\mathbf{G}}_t)$ equals to $\mathbf{q}_t$ defined in Algorithm \ref{Alg:FTPL-bandit-setting}.

\begin{lemma}\label{lem:bandit-FTPL}
	The learning regret for Algorithm~\ref{Alg:FTPL-bandit-setting} is $\tilde{\O}(\sigma+\frac{TK}{\sigma}+\varepsilon K) +\O(T\alpha)$.
\end{lemma}
\begin{proof}
	Combine Lemma~\ref{lemma:gbpa} and Eq.(\ref{eq:bandit-learning-regret-framework}), and the convexity of maximization, the learning regret is upper bounded by $\E[\max_i\bar{G}_{Ti}-\Phi(\bar{\mathbf{G}}_T)+ \Phi(\mathbf{0}) + \sum_{t\in [T]} D_{\Phi}(\bar{\mathbf{G}}_t,\bar{\mathbf{G}}_{t-1})]+2\alpha T$.
	
	By the same argument in Lemma~\ref{lem:twofold-tree-regret}, we conclude that $\max_i\bar{G}_{Ti}\le \Phi(\bar{\mathbf{G}}_T)$ and $\Phi(\mathbf{0})= \mathbb{E}_{\mathbf{u}} [\max_i u_i] = \tilde{\O}(\sigma)$.
We then need to bound the Bregman divergence. For all $t$, we have
\begin{align*}
\textstyle
D_{\Phi}(\bar{\mathbf{G}}_t,\bar{\mathbf{G}}_{t-1}) &\leq |\langle\bar{\mathbf{g}}_t,\nabla\Phi(\bar{\mathbf{G}}_t)-\nabla\Phi(\bar{\mathbf{G}}_{t-1})\rangle| = |\langle\bar{\mathbf{g}}_t,\mathbf{q}_{t}-\mathbf{q}_{t-1}\rangle| ~.
\end{align*}

If we fix the randomness of the first $t-1$ round, then $\tilde{\mathbf{q}}_{t-1}$ is determined. Denote $\mathbf{q}_{t}|_{i}$ as the result of $\mathbf{q}_{t}$ after expert $i$ is chosen in round $t$. 
Recall that for any $t\in[T],i\in[K]$, $\bar{g}_{ti}=\frac{g_{ti}}{\tilde{q}_{(t-1)i}}$ with probability $\tilde{q}_{(t-1)i}$, and $0$ otherwise. This implies $q_{ti'}|_i-q_{(t-1)i'}\le 0$ if expert $i'$ isn't chosen in round $t$, i.e. $i'\neq i$.
Take expectation over the randomness of expert chosen in round $t$ we have
\[
\textstyle
\E_{i}[|\langle\bar{\mathbf{g}}_t,\mathbf{q}_{t}-\mathbf{q}_{t-1})\rangle|\big|\tilde{\mathbf{q}}_{t-1}]\le \sum_{i \in [K]} {q}_{(t-1)i}\frac{{g}_{ti}}{\tilde{q}_{(t-1)i}} \left({q}_{ti}|_{i}-{q}_{(t-1)i}\right) \le \sum_{i \in [K]}2 {g}_{ti} \left({q}_{ti}|_{i}-{q}_{(t-1)i}\right)~,
\]
where the last inequality is because $\frac{q_{(t-1)i}}{\tilde{q}_{(t-1)i}}\le \frac{1}{1-\alpha}\le 2$. 

Now we manage to bound ${q}_{ti}|_{i}-{q}_{(t-1)i}$. 
The proof here is inspired by Lemma 6 of \cite{jain2012differentially}.
Recall that $\tilde{G}_{tj} - \bar{G}_{tj}$ follows the Gaussian distribution ${N}(0, \log T  \sigma^2)$ from Lemma~\ref{lem:tree-aggregation-noise-bandit}. 
Let $\tilde{\mathbf{G}}_t-\bar{\mathbf{G}}_t=\mathbf{u}_t$. 
Denote $\eta=\frac{1}{\sigma\sqrt{\log T}}$, 
for all $\mathbf{z}\in \mathbb{R}^K$, we have

\begin{equation}
\textstyle
\begin{split}
\textstyle
\frac{\pdf(\bar{\mathbf{G}}_t+\mathbf{u}_t=\mathbf{z})}{\pdf(\bar{\mathbf{G}}_{t-1}+\mathbf{u}_{t-1}=\mathbf{z})}&=\exp\left(\frac{\eta^2}{2}(||\bar{\mathbf{G}}_{t-1}-\mathbf{z}||^2_2-||\bar{\mathbf{G}}_t-\mathbf{z}||^2_2)\right) \\
&=\exp\left(\frac{\eta^2}{2}(||\bar{\mathbf{g}}_t||_2^2-2\langle \bar{\mathbf{g}}_t, \bar{\mathbf{G}}_t-\mathbf{z}\rangle)\right)\\
&\leq \exp\left(\frac{\eta^2}{2}(||\bar{\mathbf{g}}_t||_2^2+2|\langle\bar{\mathbf{g}}_t,\bar{\mathbf{G}}_t-\mathbf{z} \rangle|)\right) ~.\label{eq:ratio}
\end{split}
\end{equation}

Note that $\langle\bar{\mathbf{g}}_t,\bar{\mathbf{G}}_t-\mathbf{z} \rangle\sim {N}(0,||\bar{\mathbf{g}}_t||^2_2 \eta^{-2}\mathbf{I})$, then by Mill's inequality\footnote{
	Mill's inequality: For a random variable $v\sim {N}(0,1)$, and $\forall \gamma>1$, $\Pr(|v|\geq\gamma)\leq \exp(-\gamma^2/2)$.
}, we get
\begin{align}
\label{eq:Mill}
\textstyle
\Pr\left[|\langle\bar{\mathbf{g}}_t,\bar{\mathbf{G}}_t-\mathbf{z} \rangle|\geq ||\bar{\mathbf{g}}_t||_2\gamma\eta^{-1}\right] \leq \exp (-\frac{\gamma^2}{2}) ~.
\end{align}

Set $\gamma=\sqrt{2\ln\frac{1}{\delta}}$.
Put Eq.(\ref{eq:ratio}) and Eq.(\ref{eq:Mill}) together we have 
\[
\forall \mathbf{R}\subseteq \mathbb{R}^K,\quad \pr[\bar{\mathbf{G}}_t+\mathbf{u}_t\in \mathbf{R}] \leq e^{2\eta\gamma ||\bar{\mathbf{g}}_t||_2} \pr[\bar{\mathbf{G}}_{t-1}+\mathbf{u}_{t-1}\in \mathbf{R}]+\delta ~.
\] 


Let $S_i=\{\mathbf{y}\in \mathbb{R}^K:i=\argmax_jy_j\}$, then ${q}_{(t-1)i}= \pr[\bar{\mathbf{G}}_{t-1}+\mathbf{u}_{t-1}\in S_i]$, and note that ${q}_{ti}|_i =\pr[\bar{\mathbf{G}}_{t}+\mathbf{u}_{t}\in S_i]$ for some $\bar{\mathbf{g}}_t$. 
Thus, there is
\begin{align*}
\textstyle
{q}_{ti}|_{i}-{q}_{(t-1)i}\le 4\eta\gamma||\bar{\mathbf{g}}_t||_2q_{(t-1)i}+\delta=4\eta\gamma \frac{g_{ti}}{\tilde{q}_{(t-1)i}}q_{(t-1)i}+\delta \le 8\eta\gamma g_{ti}+\delta.
\end{align*}

So we have
\[
\textstyle
\E_{i_t}[|\langle\bar{\mathbf{g}}_t,\mathbf{q}_{t}-\mathbf{q}_{t-1})\rangle|\big| \tilde{\mathbf{q}}_{t-1}]\leq \sum_{i \in [K]} 2{g}_{ti}(8\eta\gamma{g}_{ti} + \delta)\leq 16\gamma\eta K + 2\delta K ~,
\]
which implies
$\E |\langle\bar{\mathbf{g}}_t,\mathbf{q}_{t}-\mathbf{q}_{t-1})\rangle|\leq 16\eta\gamma K + 2\delta K$.
Recall that $\eta=(\sqrt{\log T} \sigma)^{-1}$ and $\delta=\varepsilon/T$, the learning regret is $\tilde{\O}(\sigma+\eta\gamma TK+\delta TK)+\O(T\alpha)=\tilde{\O}(\sigma+\frac{TK}{\sigma}+\varepsilon K) +\O (T\alpha)$.
\end{proof}

\subsection{Bounding Game-theoretic Regret}\label{subsection:single-bandit-game-regret}
In bandit setting, we still have the stability of future utility as Lemma \ref{lem:single-bidder-stability} due to Lemma \ref{lem:tree-aggregation-privacy-bandit}. We then prove the game-theoretic regret is also ${\O}(T\alpha)$. 

\begin{lemma}
\label{lem:single-bidder-deviation-bandit}
On any day with price $p$, if $|v-p| \ge \alpha$, for
\begin{itemize}
	\item $\alpha=2\tau\varepsilon$ under the assumption of large market; or
	\item $\alpha=\frac{2\varepsilon}{1-\gamma}$ under the assumption of impatient bidders,
	\end{itemize}
then any strategic bidder with value $v$ will behave truthfully, i.e., buy the good when $p\le v$, and refuse to buy otherwise. 

\end{lemma}
\begin{proof}
It is easy to see that to maximize the utility in this round, the best strategy is to behave truthfully. Any non-honest behavior will cause a loss of utility compared with truthful strategy.
\begin{enumerate}
\item If $v\geq p$ but the bidder rejects the good, then the utility loss is $v-p-0=v-p\geq0$.
\item If $v < p$ but the bidder accepts the good, then the utility loss is $0-(v-p)=p-v>0$.
\end{enumerate}
Thus the utility loss in the current round is $|v-p|$. The following proof is identical as Lemma \ref{lem:single-bidder-deviation}, applying Lemma \ref{lem:single-bidder-stability}.
\end{proof}



\begin{corollary} 
	\label{cor:single-bidder-gt-regret-bandit}
	The game-theoretic regret is bounded by $\alpha T$ for
	\begin{itemize}
	\item $\alpha=2\tau\varepsilon$ under the assumption of large market; or
	\item $\alpha=\frac{2\varepsilon}{1-\gamma}$ under the assumption of impatient bidders.
	\end{itemize}
\end{corollary}

\begin{proof}
	Suppose $p^*$ is the optimal fixed price w.r.t. truth bidders. 
	Consider the revenue of fixed price $p^* - \alpha$ w.r.t. strategic bidders. 
	We call them strategic bidders in the sense that their behaviors follow Lemma \ref{lem:single-bidder-deviation-bandit}.
	Every time that price $p^*$ has a sale w.r.t. truth bidders, price $p^* - \alpha$ would also get a sale w.r.t. strategic bidders. 
	Hence, the revenue of price $p^* - \alpha$ w.r.t. strategic bidders is at least the optimal w.r.t. truth bidders minus $\alpha T$.
	So the corollary follows. 
\end{proof}

\subsection{Bounding Regret (Proof of Theorem \ref{thm:bandit_total_regret})}\label{subsection:single-bandit-total-regret}
	
\begin{proof}
	We prove the theorem under the large market assumption, and the other is similar. 
	The learning regret is $\tilde{\O}(\sigma+\frac{TK}{\sigma}+\varepsilon K +T\alpha)$ (Lemma \ref{lem:bandit-FTPL}), and game-theoretic regret is $\alpha T$ (Corollary \ref{cor:single-bidder-gt-regret-bandit}). 
	Under large market assumption, we have $\alpha=2\tau\varepsilon$.
	Recall 
	$\sigma=\tilde{\O}(\frac{K}{\alpha\varepsilon})=\tilde{\O}(\tau\alpha^{-3})$,
	the total regret is
	$\tilde{\O}(\tau\alpha^{-3}+\frac{T\alpha^2}{\tau}+\frac{1}{\tau}) +\O(T\alpha)=\tilde{\O}(\tau\alpha^{-3})+\O(T\alpha)$.
	It is at most ${\O}(T\alpha)$ when $T\geq \tilde{\O}(\tau\alpha^{-4})$.
\end{proof}

\bigskip

\end{document}